\documentclass[11pt]{article}%
\usepackage{amsmath}
\usepackage{amsfonts}
\usepackage{amssymb}
\usepackage{graphicx}
\usepackage{fullpage}
\usepackage{hyperref}%
\usepackage{bm}
\usepackage{physics}
\usepackage{algorithm}
\usepackage{algorithmic}
\usepackage{cleveref}
\providecommand{\U}[1]{\protect\rule{.1in}{.1in}}
\newtheorem{theorem}{Theorem}

\newtheorem{claim}[theorem]{Claim}

\newtheorem{corollary}[theorem]{Corollary}

\newtheorem{lemma}[theorem]{Lemma}

\newtheorem{problem}[theorem]{Problem}

\newenvironment{proof}[1][Proof]{\noindent\textbf{#1.} }{\ \rule{0.5em}{0.5em}}

\begin{document}
\title{Learning Distributions over Quantum Measurement Outcomes}
\author{Weiyuan Gong \thanks{Corresponding author. Tsinghua University. \ Email: gongwy19@mails.tsinghua.edu.cn} 
\and Scott Aaronson \thanks{University of Texas at Austin.}}
\date{}
\maketitle

\begin{abstract}
\textit{Shadow tomography} for quantum states provides a sample efficient approach for predicting the properties of quantum systems when the properties are restricted to expectation values of $2$-outcome POVMs. However, these shadow tomography procedures yield poor bounds if there are more than 2 outcomes per measurement. In this paper, we consider a general problem of learning properties from unknown quantum states: given an unknown $d$-dimensional quantum state $\rho$ and $M$ unknown quantum measurements $\mathcal{M}_1,...,\mathcal{M}_M$ with $K\geq 2$ outcomes, estimating the probability distribution for applying $\mathcal{M}_i$ on $\rho$ to within total variation distance $\epsilon$. Compared to the special case when $K=2$, we need to learn unknown distributions instead of values. We develop an online shadow tomography procedure that solves this problem with high success probability requiring $\tilde{O}(K\log^2M\log d/\epsilon^4)$ copies of $\rho$. We further prove an information-theoretic lower bound that at least $\Omega(\min\{d^2,K+\log M\}/\epsilon^2)$ copies of $\rho$ are required to solve this problem with high success probability. Our shadow tomography procedure requires sample complexity with only logarithmic dependence on $M$ and $d$ and is sample-optimal for the dependence on $K$.
\end{abstract}

% \begin{keywords}
% quantum computing, online learning,  state tomography, distribution learning, sample complexity,
% \end{keywords}

\section{Introduction}\label{sec:intro}

The statistical learning theory problem of extracting information based on empirical observations is of fundamental importance in a number of fields. In quantum physics, a fundamental problem is to obtain the properties of a quantum system based on statistical results from quantum measurements. A general method to obtain the full information of an unknown $d$-dimensional quantum state $\rho$, called \textit{quantum state tomography}, completely recovers the density matrix to within a small error. This task is proved to require $\Omega(d^2)$ copies of $\rho$ from the information-theoretical perspective \cite{Haah2017Sample}. The sample optimal procedures with exactly $O(d^2)$ sample complexity have been recently developed by O’Donnell and
Wright \cite{Odonnel2016Efficient} and Haah et al.\ \cite{Haah2017Sample}. However, these state tomography procedures of $\rho$ have been pushed to the limit of their capabilities after the recent advances in experimental quantum platforms \cite{Preskill2018Quantum}, recalling that the dimension of the quantum state $d=2^n$ increases exponentially in the number of qubits. 

On the other hand, demanding full description of a quantum state may be excessive for concrete quantum problems. Following this conceptual different line of research, the \textit{quantum shadow tomography} problem developed by Aaronson \cite{Aaronson2019Shadow,Aaronson2007Learnability} considers the case when we are given an unknown $d$-dimensional quantum state and $M$ known quantum events $E_1,...,E_M\in\mathbb{C}^{d\times d}$ with $0\preceq E_i\preceq \mathbb{I}$. Each quantum event can be regarded as a two-outcome quantum measurements that outputs $1$ (or ``accept") with probability $\Tr(E_i\rho)$ and outputs $0$ (or ``reject") otherwise. The goal is to estimate each expectation $\mathbb{E}_\rho[E_i]=\Tr(E_i\rho)$ to within additive error $\pm\epsilon$. This shadow tomography problem for two-outcome POVMs is a quantum analogue of the classical \textit{adaptive data analysis} \cite{Dwork2015Preserving}, which can be solved with $\text{poly}(\log M,\log d,1/\epsilon)$ samples \cite{Bassily2021Algorithmic}. By combining a gentle search routine with online learning algorithm, Aaronson et al.\ proved that the shadow tomography problem can be solved using $\tilde{O}(\log^4 M\log d/\epsilon^4)$ copies of $\rho$ \cite{Aaronson2019Shadow,Aaronson2018Online}, where $\tilde{O}$ hides a $\text{poly}(\log\log M,\log\log d,\log(\frac1\epsilon))$ factor. Inspired by the techniques from the field of differential privacy \cite{Dwork2010Boosting}, an alternative sample complexity of $\tilde{O}(\log^2 M\log^2 d/\epsilon^8)$ was obtained by Aaronson and Rothblum \cite{Aaronson2019Gentle}. Recently, B{\u{a}}descu and O'Donnell \cite{Buadescu2021Improved} proved the best known upper bound of sample complexity on this problem as
\begin{align*}
N=\tilde{O}\left(\frac{\log^2 M\log d}{\epsilon^4}\right).
\end{align*}
This complexity was obtained by combining a \textit{quantum threshold search} procedure and an online learning setting.

In quantum mechanics, the prediction of some intriguing properties requires quantum measurements with $K>2$ measurement outcomes. In this case, the measurement $\mathcal{M}$ outputs results $j=1,...,K$ with probability $\mathbb{E}_\rho[E_j]=\Tr(E_j\rho)$, which is expectations of quantum events $E_1,...,E_K$ that satisfies $\sum_{j=1}^KE_j=\mathbb{I}$. Our goal is to approximate the probability distribution over the outcomes of $\mathcal{M}$ within total variation distance $\epsilon$. Recall that $K$ can be close to $d$ in some cases, so it is natural to ask if we can derive an algorithmic upper bound that use fewer copies of $\rho$.

\subsection{Our Results}\label{sec:introResult}
Motivated by the problem above, this paper studies the shadow tomography problem of $K$-outcome quantum measurements, which can be formulated as follows
\begin{problem}\label{prob:KShadTom}
(Shadow Tomography of $K$-outcome Measurements) We consider an unknown $d$-dimensional quantum state, as well as $M$ quantum measurements $\mathcal{M}_1,...,\mathcal{M}_M$, each of which has $K$ results and outputs the $j$-th result with probability $\Tr(E_{i,j}\rho)$ for $i\in[M]$ and $j\in[K]$. We denote $\bm{p}_i$ the probability distribution $(\Tr(E_{i,1}\rho),...,\Tr(E_{i,K}\rho))$ after measurement $\mathcal{M}_i$. Our goal is to output $M$ probability distributions $\bm{b}_1,...,\bm{b}_M$ defined on the $K$-outcomes such that the total variance distance $d_{TV}(\bm{p_i},\bm{b}_i)\leq\epsilon$ with success probability at least $1-\delta$. 
\end{problem}

We remark that the quantum events $E_{i,j}$ for $j\in[K]$ corresponding to the quantum measurement $\mathcal{M}_i$ is defined to satisfy the constraint $0\preceq E_{i,j}\preceq\mathbb{I}$ and $\sum_{j=1}^KE_{i,j}=\mathbb{I}$. The first main result of this paper is to propose an algorithm to solve this shadow tomography problem of $K$-outcome measurements. We prove the following sample-complexity upper bound for our algorithm.
\begin{theorem}\label{thm:MainUpper}
Problem~\ref{prob:KShadTom} (Shadow Tomography of $K$-outcome Measurements) is solvable using
\begin{align*}
N=\tilde{O}\left(\frac{\log(1/\delta)}{\epsilon^4}\cdot K\cdot\log^2M\cdot\log d\right)
\end{align*}
copies of $\rho$. Here, the $\tilde{O}$ hides a $\text{poly}(\log\log M,\log\log D,\log(1/\epsilon),\log K)$ factor. The procedure is fully explicit and online. 
\end{theorem}

We provide an overview of proof for this theorem in \Cref{sec:introTech}. The detailed proof is technically involved and provided in \Cref{sec:thres} and \Cref{sec:shadow}. \Cref{thm:MainUpper} indicates that we can learn the probability distribution of $M$ quantum measurements of $K$ outcomes using sample complexity that depends logarithmically on $M$ and $d$ but linearly on $K$. Considering the parameters $M$, $d$, and $\epsilon$, our algorithm has the same dependence compared with the best known upper bound for $2$-outcome case \cite{Buadescu2021Improved}. The dependence on $K$ is the most important result in this work. Compared to directly regarding each quantum event $E_{i,j}$ as a two-outcome quantum measurement and approximating the expectation $\Tr(E_{i,j}\rho)$ to within additive error $2\epsilon/K$, our algorithm reduces the dependence on $K$ from $\tilde{O}(K^4)$ to $\tilde{O}(K)$. Notice that in some extreme cases $K$ can be as large as $\Theta(d)$, which is exponential in system size $n$, our algorithm reduces the number of copies required to perform the shadow tomography task effectively. Although the complexity of our algorithm still has an $\tilde{O}(K)$ dependence on $K$, we emphasize that this dependence is necessary by the following information-theoretic lower bound on Problem~\ref{prob:KShadTom}:
\begin{theorem}\label{thm:MainLower}
Any strategy for Problem~\ref{prob:KShadTom}---i.e., for estimating all $\bm{p}_i=(\Tr(E_{i,1}\rho),...,\Tr(E_{i,K}\rho))$ of $\mathcal{M}_i$ to within total variation distance $\epsilon$ for all $i\in[M]$, with success probability at least (say) $2/3$---requires at least
\begin{align*}
N\geq\Omega\left(\frac{\min\{d^2,K+\log M\}}{\epsilon^2}\right)
\end{align*}
copies of unknown $d$-dimensional quantum state $\rho$.
\end{theorem}

We provide the sketch of proof for this lower bound in \Cref{sec:introTech} and leave the detailed proof in \Cref{sec:lower}. The lower bound is obtained by an information theory argument developed by Flammia et al.\ \cite{Flammia2012Quantum} and refined by further works \cite{Aaronson2019Shadow,Huang2020Predicting,Haah2017Sample}. The proof exploits Holevo's theorem \cite{Holevo1973Bounds} and Fano's inequality \cite{Fano1949Transmission}. Even in the special case where there is only $M=1$ entirely classical measurement and the unknown quantum state is also classical, learning a distribution on $[K]$ to within total variance distance $\epsilon$ still requires $O(K/\epsilon^2)$ samples \cite{Canonne2020Short}. This result can be understood as the information required to approximate the probability distribution scales linearly with dimension of the distribution. By comparing the lower bound in \Cref{thm:MainLower} and the upper bound in \Cref{thm:MainUpper}, we can conclude that our algorithm for shadow tomography of $K$-outcome quantum measurement is optimal concerning the dependence on $K$.

\subsection{Techniques}\label{sec:introTech}
Our shadow tomography procedure involves combining two ideas: solving a quantum distribution threshold search problem using $O(K\log^2M/\epsilon^2)$ in each iteration and performing an online learning procedure that lasts at most $O(\log d/\epsilon^2)$ such iterations. 

The first step in this work concerns a problem we call the \textit{quantum distribution threshold search} problem, which reduces to the quantum \textit{threshold search} problem \cite{Buadescu2021Improved} in the special case when $K=2$. We formulate this problem as below:
\begin{problem}\label{prob:QThresK}
(Quantum Distribution Threshold Search) Suppose we are given
\begin{itemize}
\item Parameters $0<\epsilon,\delta<\frac12$;
\item Unentangled copies of an unknown $d$-dimensional quantum state $\rho$.
\item A list of $M$ $d$-dimensional POVMs $\mathcal{M}_1,...,\mathcal{M}_M$ each of $K$ outcomes corresponding to quantum events $E_{i,j}$, where $i\in[M]$, $j\in[K]$, and $\sum_{j=1}^KE_{i,j}=\mathbb{I}$. We denote $\bm{p}_i=(\Tr(E_{i,1}\rho),...,\\\Tr(E_{i,K}\rho))$ to be the actual distribution over the measurement outcomes of $\mathcal{M}_i$.
\item A list of $M$ threshold vectors $\bm{\theta}_i=(\theta_{i,1},...,\theta_{i,K})$, for $\theta_{i,j}\in[0,1]$ and $\sum_{j=1}^K\theta_{i,j}=1$.
\end{itemize}
the algorithm outputs either:
\begin{itemize}
\item $d_{TV}(\bm{p}_{i^*},\bm{\theta}_{i^*})>3\epsilon/4$ for some particular $i^*$; or
\item $d_{TV}(\bm{p}_{i},\bm{\theta}_{i})\leq\epsilon$ for all $i$.
\end{itemize}
Our goal is to minimize the number of copies required to ensure we output correctly with success probability at least $1-\delta$. 
\end{problem}

This problem for the case of $K=2$ was originally called a gentle search procedure in Ref. \cite{Aaronson2019Shadow}. Later, it was renamed as the quantum threshold problem in Ref. \cite{Buadescu2021Improved} since a \textit{gentle} measurement assumption \cite{Aaronson2019Gentle} is not necessary. It is proved that the quantum threshold problem can be solved using $\tilde{O}(\log^2 M/\epsilon^2)$ copies of $\rho$ with probability at least (say) $3/4$ \cite{Buadescu2021Improved}. In this paper, we provide an algorithm that can solve Problem~\ref{prob:QThresK} for any $K\geq 2$:
\begin{theorem}\label{thm:QThresKUpper}
Problem~\ref{prob:QThresK} (Quantum Distribution Threshold Search) is solvable using
\begin{align*}
N=\tilde{O}\left(\frac{\log(1/\delta)}{\epsilon^2}\cdot K\cdot\log^2M\right)
\end{align*}
copies of $\rho$.
\end{theorem}

We provide the proof of this theorem in \Cref{sec:thres}. When $K=2$, our upper bound for quantum distribution threshold search problem reduces to the same bound for quantum threshold search problem. We remark that the $K$ dependence in the sample complexity bound we provide in \Cref{thm:MainUpper} directly comes from the $K$ dependence in solving the quantum distribution threshold search problem. 

Given our quantum distribution threshold search algorithm, the second step is to employ a black-box reduction to an online quantum state learning algorithm. At the special case when $K=2$, the bound is obtained by Aaronson et al.\ \cite{Aaronson2018Online}. The formal version of our result in online learning distributions is provided as follows:
\begin{theorem}\label{thm:OnlineUpper}
Let $\rho$ be an unknown $d$-dimensional quantum state, as well as $\mathcal{M}_1,\mathcal{M}_2,...,\mathcal{M}_t,...$ be a sequence of $K$-outcome POVMs each consisting quantum events $E_{t,j}$ for $j\in[K]$. We denote $\bm{p}_t=(\Tr(E_{t,1}\rho),...,\Tr(E_{t,K}\rho))$ to be the actual probability distribution when we apply $\mathcal{M}_t$ on $\rho$. We are provided with a probability distribution $\bm{b}_t$ after each measurement $\mathcal{M}_t$ such that $d_{TV}(\bm{p}_t,\bm{b}_t)\leq\epsilon/4$. There exists a strategy for outputting hypothesis states $\omega_1,\omega_2,...$ such that the probability distribution $\bm{\mu}_t$, which is obtained by applying $\mathcal{M}_t$ on $\omega_t$, deviates more than $3\epsilon/4$ from $\bm{p}_t$ for at most $T=O(\log d/\epsilon^2)$ iterations $t$ (also called a ``bad iteration").
\end{theorem}

We provide the proof for \Cref{thm:OnlineUpper} in \Cref{sec:shadowOnline} following the template of the \textit{Regularized Follow-the-Leader algorithm} (RFTL; see, for example in Hazan et al.\ \cite{Hazan2016Introduction}). We can then combine our quantum distribution threshold search algorithm with this online setting to prove the sample complexity for our shadow tomography procedure of $K$-outcome quantum measurements. We start with the maximally mixed state  $\mathbb{I}/d$. In each iteration, we first perform the quantum distribution threshold search algorithm to find an $i^*$ such that the total variance distance between $\bm{\mu}_t$ and $\bm{p}_t$ is larger than $3\epsilon/4$. We can then use $\tilde{O}(K/\epsilon^2)$ samples to estimate $\bm{b}_t$ with high success probability and update the hypothesis. As there are at most $O(\log d/\epsilon^2)$ ``bad iterations", we finally reach the complexity bound in \Cref{thm:MainUpper}.

To prove the lower bound, we first fix $M$ different quantum measurements. We then find a set (known as packing net \cite{Haah2017Sample}) of size $2^{K/2}M$ consisting of mixed states $\{\rho_1,...,\rho_N\}$ such that we can use our shadow tomography procedure to distinguish between any pair of states chosen from this packing net, which requires $\log(2^{K/2}M)=\Theta(K+\log M)$ bits of information. We further show using Holevo's theorem \cite{Nielsen2002Quantum} that we can at most obtain $O(\epsilon^2)$ bits of information from any quantum states chosen from this set. Therefore, the sample complexity is bounded below by $\Omega((K+\log M)/\epsilon^2)$ to make it possible to obtain the information.

\subsection{Related Works}\label{sec:introReview}
Here, we compare \Cref{thm:MainUpper} with some related works and show the difference and connection among these results.

\textbf{Shadow tomography of two-outcome quantum measurements.} A first related topic is the shadow tomography of two-outcome quantum measurements \cite{Aaronson2019Shadow,Buadescu2021Improved}. It is proved that only $O(\log^2 M\log d/\epsilon^4)$ copies of $\rho$ can estimate the expectation value of $M$ two-outcome quantum measurement. When $K=2$, the sample complexity in \Cref{thm:MainUpper} reduces to this bound. To extend this result to the case of $K>2$, a straightforward approach is to regard each quantum event as a two-outcome quantum measurement and estimate each expectation value $\Tr(E_{i,j}\rho)$ to within additive error $2\epsilon/K$. However, this approach requires an additional cost of $\tilde{O}(K^4)$ using the state-of-art shadow tomography algorithms. Compared to this direct extension, our approach only requires sample complexity that increases linear with $K$ and is proved to be optimal concerning the dependence on $K$.

\textbf{Quantum function estimation using classical shadow.} Huang et al.\ \cite{Huang2020Predicting} considered the task of estimating quantum functions (or the expectations of quantum operators). For completeness, let us restate that theorem in the language of this paper.
\begin{theorem}\label{thm:HuangShadow}
(Huang, Keung, and Preskill \cite{Huang2020Predicting}) Given $\rho$ an unknown $d$-dimensional state, as well as $M$ quantum operators $O_1,...,O_M$. There exists a strategy that can approximate the expectation value for each operator $\Tr(O_i\rho)$ to within additive error $\epsilon$ with high success probability (say $3/4$) using
\begin{align*}
\tilde{O}\left(\frac{\max_i\norm{O_i}_{\text{shadow}}^2}{\epsilon^2}\cdot\log M\right)
\end{align*}
copies of $\rho$, where the shadow norm $\norm{O}_{\text{shadow}}$ is defined to be
\begin{align*}
\norm{O}_{\text{shadow}}=\max_{\sigma}\left(\mathbb{E}_{U\sim\mathcal{U}}\sum_{b\in\{0,1\}^n}\bra{b}U\sigma U^\dagger\ket{b}\bra{b}U\mathcal{M}^{-1}(O) U^\dagger\ket{b}^2\right)^{1/2}.
\end{align*}
Here, the maximization goes through all $d$-dimensional state $\sigma$ in the Hilbert space.
\end{theorem}
The algorithm for \Cref{thm:HuangShadow} employs the "classical shadow" and does not require the joint measurement that simultaneously measures states of the form $\rho^{\otimes k}$. However, the norm $\norm{O}_\text{shadow}$ is closely related to the Hilbert-Schmidt norm, which may increase exponentially in the system size $n=\log d$. This algorithm can be efficient if $O_i$'s are local operators acting on the bounded-size subsystem. Our algorithm, however, can provide sample-efficient shadow tomography when the number measurement outcome scales polynomially with system size $n$, regardless of whether the measurement is global. Our algorithm can effectively reduce the number of samples required for a shadow tomography problem for global POVMs.

\textbf{Quantum state tomography. }It is proved that there exists a sample-optimal algorithm that can perform state tomography for an unknown quantum state $\rho$ of rank $r\leq d$ using $O((dr/\epsilon^2)$ copies of $\rho$ \cite{Haah2017Sample}. Although the shadow tomography procedure in this paper does not require full information of $\rho$, the information obtained in this procedure increases linearly with $K$. We can observe the connection between quantum state tomography and our shadow tomography algorithm when $K$ becomes exponentially large and we can obtain enough information to approximate the full description of some particular type of $\rho$. In the extreme case, when we perform a quantum measurement on the computational basis---i.e., there are $d$ possible outcomes corresponding to all possible $n$-bit classical strings $\bm{x}$ chosen from $\{0,1\}^n$. The quantum events corresponding to the $\bm{x}$ are the projectors
\begin{align*}
E_{\bm{x}}=\ket{\bm{x}}\bra{\bm{x}},\forall\bm{x}\in\{0,1\}^n.
\end{align*}
To perform shadow tomography on this measurement, we require $\Omega(d/\epsilon^2)$ copies of $\rho$. By performing this measurement, we can obtain a full description of any pure states. Recall that the rank for pure states is $r=1$ and state tomography for quantum states of rank $r$ requires a sample complexity of $\Omega(dr/\epsilon^2)$. This bound is the same as the sample complexity required for state tomography for pure states. 

\section{Preliminaries}\label{sec:prelim}

\subsection{Classical Probability Theory}\label{sec:prelimProb}
Consider two probability distributions $\mathcal{D}=(p_x)_x$ and $\mathcal{D}'=(q_x)_x$ on $K$-dimensional space, we will use the following three distance measures between them. The \textit{total variance distance} between $\mathcal{D}$ and $\mathcal{D}'$ is defined by
\begin{align*}
d_{TV}(\mathcal{D},\mathcal{D}')=\frac{1}{2}\sum_x\abs{p_x-q_x}.
\end{align*}
We also consider the two distance measures that are commonly used for vectors. The \textit{Euclidean norm} of the distance between the two distributions is defined by
\begin{align*}
\norm{\mathcal{D}-\mathcal{D'}}_2=\left(\sum_x(p_x-q_x)^2\right)^{1/2}.
\end{align*}
The \textit{infinity norm} of $(\mathcal{D}-\mathcal{D}')$ is defined by
\begin{align*}
\norm{\mathcal{D}-\mathcal{D'}}_\infty=\max_x\abs{p_x-q_x}.
\end{align*}

The Euclidean norm and the infinity norm are not commonly used in probability theory. We employ these vector norms as intermediate tools when using the concentration inequalities on random vectors. To connect among these norms, we notice that for any probability distribution $\mathcal{D}$ and $\mathcal{D}'$, the following inequality holds
\begin{align}\label{eq:DisBound}
d_{TV}(\mathcal{D},\mathcal{D}')\leq\frac{\sqrt{K}}{2}\norm{\mathcal{D}-\mathcal{D'}}_2.
\end{align}

\subsection{Concentration of Random Vectors}\label{sec:prelimVecbern}
We will need a concentration lemma for random vectors, which is an extension of the vector Bernstein inequality (Theorem 6) in Ref. \cite{Gross2011Recovering,Kohler2017Sub}.
\begin{lemma}\label{lem:VecBerstein}
Let $\mathbf{x}_1,...,\mathbf{x}_m$ be independent $K$-dimensional vector-valued random variables. We assume that each random vector is zero-mean, uniformly bounded and has bounded variance, i.e., 
\begin{align*}
\mathbb{E}[\mathbf{x}_i]=0\text{ and }\norm{\mathbf{x}_i}_\infty\leq\mu\text{ as well as }\mathbb{E}\left[\norm{\mathbf{x}_i}_2^2\right]\leq\sigma^2 
\end{align*}
for some constants $\mu,\sigma>0$. Suppose that parameters satisfies $0<\epsilon<\sigma^2/\mu$, then we have
\begin{align*}
\Pr\left\{\norm{\frac{1}{m}\sum_{i=1}^m\mathbf{x}_i}_2\geq\epsilon\right\}\leq\exp(-m\cdot\frac{\epsilon^2}{8\sigma^2}+C),
\end{align*}
for some positive constant $C$.
\end{lemma}
\begin{proof}
Theorem $6$ in Ref. \cite{Gross2011Recovering} indicates that for independent, zero-mean random vectors
\begin{align*}
\Pr\left\{\norm{\sum_{i=1}^m\mathbf{x}_i}\geq t+\sqrt{V}\right\}\leq\exp(-\frac{t^2}{4V}),
\end{align*}
where $V=\sum_{i=1}^m\mathbb{E}\left[\norm{\mathbf{x}_i}^2_2\right]$ is the sum of variances for random vectors. We define $\epsilon=t+\sqrt{V}$ and rewrite the above inequality as 
\begin{align*}
\Pr\left\{\norm{\sum_{i=1}^m\mathbf{x}_i}\geq\epsilon\right\}\leq\exp(-\frac14\left(\frac{\epsilon}{\sqrt{V}}-1\right)^2)\leq\exp(-\frac{\epsilon^2}{8V}+\frac14).
\end{align*}
Since the sum of variance $V$ can be bounded by $m\sigma^2$ according to our assumption, we can finally obtain the following inequality
\begin{align*}
\Pr\left\{\norm{\frac1m\sum_{i=1}^m\mathbf{x}_i}\geq\epsilon\right\}\leq\exp(-\frac14\left(\frac{\epsilon}{\sqrt{V}}-1\right)^2)\leq\exp(-m\cdot\frac{\epsilon^2}{8\sigma^2}+\frac14).
\end{align*}
By choosing the constant $C=\frac14$, we finish the proof for this lemma.
\end{proof}

Now we consider sampling from a probability distribution $\bm{p}=(p_1,...,p_K)$ for $m$ times. For the $i$-th sample where $i\in[m]$, we obtain one sample $\bm{\hat{p}}_i=(\hat{p}_i^1,...,\hat{p}_i^K)$ with only one entry $1$ and the other entries $0$. We set $\mathbf{x}_i=\bm{p}-\bm{\hat{p}}_i$. Then $\mathbf{x}_i$ is centered because $\mathbb{E}(\mathbf{x}_i)=\mathbb{E}[\bm{p}-\bm{\hat{p}}_i]=0$. Each entry of $\mathbf{x}_i$ is bounded below by $1$ and
\begin{align*}
\mathbb{E}[\norm{\mathbf{x}_i}^2]=\sigma^2=1-\sum_{j=1}^K p_j^2<1.
\end{align*}
Therefore, by Lemma~\ref{lem:VecBerstein}, we can guarantee that
\begin{align*}
\Pr\left(\norm{\frac{1}{m}\sum_{i=1}^m\bm{\hat{p}}_i-\bm{p}}_2\geq\epsilon\right)\leq\delta,
\end{align*}
as long as we choose $m\geq O(\log(1/\delta)/\epsilon^2)$. To bound the total variance distance $d_{TV}(\frac{1}{m}\sum_{i=1}^m\bm{\hat{p}}_i,\bm{p})$ between the empirical distribution and the actual distribution, we combine the bound in Eq.~\eqref{eq:DisBound} with Lemma~\ref{lem:VecBerstein} and obtain:
\begin{align*}
\Pr\left(d_{TV}\left(\frac{1}{m}\sum_{i=1}^m\bm{\hat{p}}_i,\bm{p}\right)\geq\epsilon\right)\leq\Pr\left(\norm{\frac{1}{m}\sum_{i=1}^m\bm{\hat{p}}_i-\bm{p}}_2\geq\frac{2\epsilon}{\sqrt{K}}\right)\leq\exp(-m\cdot\frac{4\epsilon^2}{K}+\frac 14).
\end{align*}
Hence, we can bound $d_{TV}(\frac{1}{m}\sum_{i=1}^m\bm{\hat{p}}_i,\bm{p})$ below $\epsilon$ with probability at least $1-\delta$ if 
\begin{align}\label{eq:TotVarDisSampBound}
m\geq O\left(\frac{\log(1/\delta)}{\epsilon^2}\cdot K\right).
\end{align}

\subsection{States, Distance Measure and Measurements in Quantum Information}\label{sec:prelimQbase}
Here, we briefly review some basic notations and concepts in quantum information. More details can be found, for example, in Nielsen and Chuang \cite{Nielsen2002Quantum}. 

A matrix $A\in\mathbb{C}^{d\times d}$ is said to be a \textit{Hermitian} matrix if $A^\dagger=A$, where $A^\dagger$ denotes the conjugate transpose of $A$. We write $A\succeq 0$ to denote that the Hermitian operator $A$ is positive semidefinite. We write $A\succeq B$ to denote $A-B\succeq0$. We use $\mathbb{I}$ for the identity matrix and the dimension can be understood from the context.  

In quantum mechanics, a $d$-dimensional \textit{quantum state} can be written as a matrix $\rho\in\mathbb{C}^{d\times d}$ with $\rho\succeq 0$ and $\Tr(\rho)=1$. If $\rho$ has rank $1$, it is called a \textit{pure state} and can be written as a outer product $\ket{\psi}\bra{\psi}$ of a complex vector $\ket{\psi}$. Equivalently, we can write $\rho$ as a convex combination for outer products of different pure states (without loss of generality, there can be at most $d$ orthogonal pure states):
\begin{align*}
\rho=\sum_{i=1}^dp_i\ket{\psi_i}\bra{\psi_i},
\end{align*}
where $\sum_{i=1}^dp_i=1$ and $p_i\geq 0$ for arbitrary $i\in[d]$. This representation can be interpreted as a probability distribution over each pure state $\ket{\psi_i}$. In the special case when $\rho$ is diagonal, it represents a classical probability distribution over orthogonal computational basis $\ket{1},...,\ket{d}$. The \textit{maximally mixed state} $\mathbb{I}/d$ corresponds to the uniform distribution over $\ket{1},...,\ket{d}$.

A \textit{quantum observable}, or \textit{quantum operator}, is a $d$-dimensional Hermitian matrix $O\in\mathbb{C}^{d\times d}$. A quantum observable is a real-value property of the physical systems. Given a quantum state $\rho$, the \textit{expectation} of $O$ with respect to $\rho$ is defined by
\begin{align*}
\mathbb{E}_\rho[O]=\Tr(O\rho).
\end{align*}
A \textit{quantum event} is a quantum operator that satisfies $0\preceq E\preceq \mathbb{I}$, i.e., a Hermitian operator with eigenvalues chosen from $[0,1]$. The expectation value of a quantum event $\mathbb{E}_\rho[E]$ can be interpreted as a probability assigned by quantum state $\rho$ to $E$. We further call $E$ a \textit{projector} for the special case when $E^2=E$ and all eigenvalues for $E$ are Boolean values $0$ and $1$.

A \textit{quantum measurement} $\mathcal{M}$, or a \textit{positive-operator valued measure (POVM)}, is a sequence $\mathcal{M}=(E_1,...,E_K)$ of quantum events with $\sum_{j=1}^KE_j=\mathbb{I}$. According to the linearity of trace, we can obtain
\begin{align*}
\sum_{j=1}^K\mathbb{E}_\rho[E_j]=\mathbb{E}_\rho\left[\sum_{j=1}^KE_j\right]=\mathbb{E}_\rho[\mathbb{I}]=1.
\end{align*}
Given a quantum state $\rho$, a POVM determines a probability distribution $\mathcal{D}=\{p_j\}_{j}$ on $[K]$ defined by $p_j=\mathbb{E}_{\rho}[E_j]=\Tr(E_j\rho)$. Mathematically, an implementation of $\mathcal{M}$ is a sequence of matrices $M_1,...M_K$ with $M_j^{\dagger}M_j=E_j$ for $j\in[d]$. Conditioned on outcome $j$, the quantum state $\rho$ collapses to the new state $\rho|_{M_j}$ defined by
\begin{align*}
\rho|_{M_j}=\frac{M_j\rho M_j^\dagger}{\mathbb{E}_\rho[M_j^\dagger M_j]}=\frac{M_j\rho M_j^\dagger}{\mathbb{E}_\rho[E_j]}.
\end{align*}
More generally, we can represent a quantum measurement as a quantum operation. A \textit{quantum operation} $S$ is defined by $d$-column matrices $M_1,...,M_K$ with $\sum_{j=1}^KM_j^\dagger M_j\preceq\mathbb{I}$. The result of applying $S$ to $\rho$ without knowing the outcome (which can be regarded as a post-selection) is
\begin{align*}
S(\rho)=\sum_{j=1}^KM_j\rho M_j^\dagger.
\end{align*}

There are a variety of distance measures between two quantum state $\rho,\sigma\in\mathbb{C}^{d\times d}$. In this work, we use the \textit{trace distance} $d_{\Tr}(\rho,\sigma)$ defined by
\begin{align*}
d_{\Tr}(\rho,\sigma)=\frac12\norm{\rho-\sigma}_{\Tr}=\max_{0\preceq E\preceq\mathbb{I}}\abs{\mathbb{E}_\rho[E]-\mathbb{E}_\sigma[E]},
\end{align*}
where $\norm{A}_{\Tr}$ is the trace norm of matrix $A$. The second inequality follows from Helstrom's theorem \cite{Helstrom1969Quantum}. In fact, for any classical distance, there is a corresponding measured quantum distance, which is the maximal classical distance that can be achieved by performing an identical measurement on $\rho$ and $\sigma$. Based on this paradigm, the quantum trace distance can be considered as a measured quantum version of total variation distance. If the trace distance between $\rho$ and $\sigma$ is bounded below by $\epsilon$, the total variation distance between the output distribution of quantum measurement $\mathcal{M}$ on $\rho$ and $\sigma$ is also upper bounded by $\epsilon$. 

\subsection{Online Learning Settings and Regrets}\label{sec:PrelimOnline}
In online learning of quantum states considered in \Cref{thm:OnlineUpper}, we are given a sequence of quantum measurements $\mathcal{M}_1,\mathcal{M}_2,...$ in each iteration $t$. In each iteration, the learner constructs a hypothesis state $\omega_t\in\mathbb{C}^{d\times d}$. Given the quantum measurement $\mathcal{M}_t$, the learner calculates the distribution after applying $\mathcal{M}_t$ on the hypothesis state $\omega_t$ as $\bm{\mu}_t=(\Tr(E_{t,1}\rho),...,\Tr(E_{t,K}\rho))$, which is known as a "prediction". 

The learner then obtains feedback from the measurement $\mathcal{M}$. The simplest feedback can be a random variable $Y_t$ chosen from value $[K]=\{1,...,K\}$ for different outcomes. In this paper, the learner obtains a feedback by performing a quantum distribution threshold search to find whether $d_{TV}(\bm{\mu}_{t},\bm{p}_{t})$ is larger than some tolerance threshold, where $\bm{p}_t=(\Tr(E_{t,1}\rho),...,\Tr(E_{t,K}\rho))$ is the actual probability distribution for the unknown state $\rho$. 

If the quantum distribution threshold search procedure does not output $t$, the learner accepts the prediction and set it as the final result. If the quantum distribution threshold search procedure outputs $t$, the learner starts an update procedure. The learner first estimates a probability distribution $\bm{b}_t$. According to Eq.~\eqref{eq:TotVarDisSampBound}, the learner can guarantee that $d_{TV}(\bm{b}_t,\bm{p}_t)\leq\epsilon/4$ with high probability by using $O(K/\epsilon^2)$ copies of $\rho$. Then, the learner defines a loss function that measures the total variance distance between the ``bad prediction" $\bm{\mu}_t$ and $\bm{b}_t$ as:
\begin{align}\label{eq:DefLoss}
\ell_t(\bm{\mu}_t):=\frac12\sum_{j=1}^K\abs{\Tr(E_{t,j}\omega_t)-b_{t,j}},
\end{align}
where $b_{t,j}$ denotes the $j$-th entry of $\bm{b}_t$. The learner updates the hypothesis $\omega_t\to\omega_{t+1}$ based on the loss, measurements, and feedback before the current iteration. 

Our goal is to design a strategy such that the learner’s total loss is minimized. Suppose there are in total $T$ iterations, we want to find a strategy such that the learner's total loss is not much more than that of the strategy which outputs the same quantum hypothesis $\varphi$ in each iteration, where $\varphi$ is chosen as the minimization of the total loss \textit{with perfect hindsight.} Formally, we define the \textit{regret} $R_T$ to be the difference between values of total loss for these two strategies as 
\begin{align}\label{eq:DefRegret}
R_T:=\sum_{t=1}^T\ell_t(\bm{\mu}_t)-\min_{\varphi\in\mathbb{C}^{d\times d}}\sum_{t=1}^T\ell_t(\bm{\mu}_{\varphi}),
\end{align}
where $\bm{\mu}_{\varphi}=(\Tr(E_{t,1}\varphi),...,\Tr(E_{t,K}\varphi))$ is the probability distribution after applying $\mathcal{M}_t$ on $\varphi$. We remark that the sequence of measurements $\mathcal{M}_t$ can be arbitrary, even adversarial, based on the learner's prior actions. 

\section{Quantum Distribution Threshold Search}\label{sec:thres}

In this section, we prove \Cref{thm:QThresKUpper}. In \Cref{sec:shadow}, we will use this procedure as a feedback in the online learning procedure of our shadow tomography algorithm of $K$-outcome POVMs. 

\subsection{Expectation Estimation}\label{sec:thresEst}
Our starting point is the following expectation estimation lemma.
\begin{lemma}\label{lem:ExpEst}
Let $\rho$ be an unknown $d$-dimensional state, and let $\mathcal{M}$ be a $K$-outcome POVM that outputs a probability distribution $\bm{p}=(\Tr(E_1\rho),...,\Tr(E_K\rho))$. We choose parameters $0<\epsilon,\delta<\frac12$. Then there exists $N=K\log(1/\delta)/\epsilon^2$ such that, for any $d$-dimensional quantum states $\rho$,
\begin{align*}
\Pr\left(d_{TV}(\bm{p},\bm{p}')\geq\frac{\epsilon}{8}\right)\leq\delta,
\end{align*}
where $\bm{p}'=(p_1',...,p_K')$ is the empirical distribution by applying $\mathcal{M}$ to the joint state $\rho^{\otimes N}$

Moreover, there exists a quantum event $B$ such that for any $K$-dimensional distribution $\bm{\tau}$
\begin{align*}
d_{TV}(\bm{p},\bm{\tau})&> \epsilon\Rightarrow\mathbb{E}_{\rho^{\otimes N}}[B]> 1-\delta,\\
d_{TV}(\bm{p},\bm{\tau})&\leq \frac{3\epsilon}{4}\Rightarrow\mathbb{E}_{\rho^{\otimes N}}[B]\leq \delta.
\end{align*}
\end{lemma}
\begin{proof}
We assign an index for each single copy $\rho$ in the joint state $\rho^{\otimes N}$ and assume each single $\rho$ occupies a ``register". For all $N$-bit classical strings $\bm{x}=(x_1,...,x_N)\in[K]^N$, we define quantum events $E_{\bm{x}}=E_{x_1}\otimes..\times E_{x_N}$ to be the tensor product of quantum event $E_{x_i}$ in the $i$-th register. It is easy to verify that $\sum_{\bm{x}\in\{0,1\}^N}E_x=\mathbb{I}$. For all $K$-dimensional positive integer arrays $\bm{k}=(k_1,...,k_K)$ with $\sum_{j=1}^Kk_j=N$, we define quantum event $A_{\bm{k}}$ to be
\begin{align*}
A_{\bm{k}}=\sum_{\bm{x}\in[K]^{\otimes N}\atop[\text{num of }x_i=j]=k_j}E_x.
\end{align*}
Then the empirical approximation $\bm{p}'$ is chosen as $\bm{p}'=\bm{k}/N$. Since each entry $k_i$ of $\bm{k}$ is distributed as Binomial$(N,\Tr(E_i\rho))$, we can bound the following probability using Eq.~\eqref{eq:TotVarDisSampBound}:
\begin{align}\label{eq:LemExpEstGua}
\Pr\left(d_{TV}(\bm{p},\bm{p}')\geq\frac\epsilon8\right)\leq\delta
\end{align}
as $N=O(K\log(1/\delta)/\epsilon^2)$.

We define a function $f:[0,1]^{\otimes K}\to\{0,1\}$ by
\begin{align*}
f(\bm{t})=\begin{cases}
&1,\text{ }d_{TV}(\bm{t},\bm{\tau})\geq \frac{7\epsilon}{8},\\
&0,\text{ otherwise}.
\end{cases}
\end{align*}
Based on this function, we define quantum event $B$ by
\begin{align*}
B&=\sum_{\bm{k}\atop k_1+...+k_K=N}f\left(\frac{\bm{k}}{N}\right) A_{\bm{k}}.
\end{align*}
As each entry $k_i$ of $\bm{k}$ is distributed as Binomial$(N,\Tr(E_i\rho))$, we can observe that
\begin{align*}
\mathbb{E}_{\rho^{\otimes N}}[B]&=\Pr\left(d_{TV}(\bm{p}',\bm{\tau})\geq \frac{7\epsilon}{8}\right).
\end{align*}
Recall the guarantee in Eq.~\eqref{eq:LemExpEstGua}. The condition $d_{TV}(\bm{p},\bm{\tau})>\epsilon$ implies that $d_{TV}(\bm{p}',\bm{\tau})\geq7\epsilon/8$ by triangle inequality. Hence,
\begin{align*}
\mathbb{E}_{\rho^{\otimes N}}[B]&=\Pr\left(d_{TV}(\bm{p}',\bm{\tau})\geq \frac{7\epsilon}{8}\right)\geq\Pr\left(d_{TV}(\bm{p},\bm{p}')\leq\frac{\epsilon}{8}\right) \geq 1-\delta.
\end{align*}
Similarly, the condition $d_{TV}(\bm{p},\bm{\tau})\leq3\epsilon/4$ implies that $d_{TV}(\bm{p}',\bm{\tau})\leq7\epsilon/8$. Hence, 
\begin{align*}
\mathbb{E}_{\rho^{\otimes N}}[\overline{B}]&=\Pr\left(d_{TV}(\bm{p}',\bm{\tau})> \frac{7\epsilon}{8}\right)>\Pr\left(d_{TV}(\bm{p},\bm{p}')\leq\frac{\epsilon}{8}\right) \geq 1-\delta.
\end{align*}
\end{proof}

Moreover, we can observe that if $E_i$'s are projectors, then $A_{\bm{k}}$'s are also projectors. Since $B$ is a summation of $A_{\bm{k}}$'s, $B$ is also a projector. By using the above lemma, we reduce the shadow tomography procedure of a $K$-outcome POVM to evaluating the expectation of a two-outcome POVM. 

\subsection{Sample Complexity of Quantum Distribution Threshold Search}\label{sec:thresComplexity}
To prove \Cref{thm:QThresKUpper}, we require the following Lemma 4.2 in Ref. \cite{Buadescu2021Improved}.
\begin{theorem}\label{thm:Buadescu2021Improved}
(B$\breve{a}$descu and O'Donnell \cite{Buadescu2021Improved}) Suppose we are given an unknown $d$-dimensional quantum state $\rho$, and $M$ quantum projectors $B_1,...,B_M\in\mathbb{C}^{d\times d}$.There exists an algorithm using $O(\log^2 M\log(1/\delta))$ copies of $\rho$ outputs either
\begin{itemize}
    \item $\mathbb{E}_{\rho}[B_{i^*}]=\Tr(B_{i^*}\rho)>1/4$ for some particular $i^*$; or
    \item $\mathbb{E}_{\rho}[B_{i}]\leq3/4$ for all $i$.
\end{itemize}
The success probability is at least $1-\delta$.
\end{theorem}

The proof of this theorem employs the $\chi^2$-stable threshold reporting technique, which is a quantum version of classical statistical results fitting into the adaptive data analysis framework. We omit the details here and refer to Ref. \cite{Smith2017NoteADA}, for example, for the related background.

Now, we begin to prove \Cref{thm:QThresKUpper}. Notice that the assumptions on $E_{i,j}$ is a quantum event in Problem~\ref{prob:QThresK} while the assumptions for $E_{i,j}$ is a projector in \Cref{thm:Buadescu2021Improved}, we have to first reduce the theorem to the case of projectors. Let $\rho\in\mathbb{C}^{d\times d}$ be the unknown quantum state and $E_{i,j}$ be the quantum events for $i\in[M]$ and $j\in[K]$. We need the following Naimark's theorem(see, for example, \cite{Riesz2012Functional,Akhiezer2013Theory}).
\begin{theorem}\label{thm:Naimark}
(Naimark) Suppose $E\in\mathbb{C}^{d\times d}$ is a quantum event, then there exists a projector $\Pi$ on the space $\mathbb{C}^{2d\times 2d}$, such that for arbitrary $\rho$,
\begin{align*}
\mathbb{E}_{\rho\otimes\ket{0}\bra{0}}[\Pi]=\mathbb{E}_{\rho}[E].
\end{align*}
\end{theorem}

We first extend the state $\rho$ to $\rho\otimes\ket{0}\bra{0}$. We can always find a projector $E'_{i,j}\in\mathbb{C}^{2d\times 2d}$ such that $\Tr(E_{i,j}\rho)=\Tr(E'_{i,j}\rho\otimes\ket{0}\bra{0})$. We remark that the state $\rho\otimes\ket{0}\bra{0}$ can be prepared even without knowing $\rho$ and the dimension of the system increases by a factor of $2$, thus making no differences to the sample complexity. Therefore, in the following, we assume that $E_{i,j}$ are projectors.

Suppose we are given $M$ $K$-outcome POVMs $\mathcal{M}_1,...,\mathcal{M}_M$ and $M$ threshold vectors $\bm{\theta}_1,...,\bm{\theta}_M$. We first apply Lemma~\ref{lem:ExpEst} with parameters $\delta=1/4$ and $\tau=\bm{\theta}_i$ for each measurement $\mathcal{M}_i$. Therefore, we can find some $N_0=O(K/\epsilon^2)$ such that each measurement $\mathcal{M}_i$ can be replaced by a quantum event $B_i\in(\mathbb{C}^{d\times d})^{\otimes N_0}$ satisfying
\begin{itemize}
    \item if $d_{TV}(\bm{p}_i,\bm{\theta}_i)>\epsilon$, $\mathbb{E}_{\rho^{\otimes N_0}}[B_i]>3/4$;
    \item if $d_{TV}(\bm{p}_i,\bm{\theta}_i)\leq3\epsilon/4$, $\mathbb{E}_{\rho^{\otimes N_0}}[B_i]\leq1/4$;
\end{itemize}
Here $\bm{p}_i$ is the actual distribution after applying $\mathcal{M}_i$ on $\rho$. Since $E_{i,j}$'s are projectors, quantum events $B_i$ are also projectors. 

We then apply \Cref{thm:Buadescu2021Improved} by setting each $B_i$ to be the projectors we have just constructed and unknown state to be $\rho'=\rho^{\otimes N_0}$. If the algorithm outputs $i^*$ such that $\mathbb{E}_{\rho'}[B_{i^*}]>1/4$, then we find $d_{TV}(\bm{p}_{i^*},\bm{\theta}_{i^*})>3\epsilon/4$. Otherwise, we can guarantee that $d_{TV}(\bm{p}_{i},\bm{\theta}_{i})\leq\epsilon$ for all $i\in[M]$ with high probability.

\section{Shadow Tomography of $K$-outcome POVMs}\label{sec:shadow}
In this section, we first prove \Cref{thm:OnlineUpper} in \Cref{sec:shadowOnline}. We then prove the first main result of our paper, \Cref{thm:MainUpper}.

\subsection{Online Learning of Quantum States}\label{sec:shadowOnline}

We suppose there are in total $T$ iterations when the learner performs an update procedure. In the update procedure, the learner follows the template of the Regularized Follow-the-Leader algorithm (RFTL) as Algorithm~\ref{algo:RFTLKShad}.
\begin{algorithm}[t] 
\caption{RFTL for Quantum Tomography of $K$-outcome POVMs}
\begin{algorithmic}[1] \label{algo:RFTLKShad}
\STATE Input: $T$, $\eta<\frac{1}{2}$
\STATE Set $\omega_1:=\mathbb{I}/d$.
\FOR{$t=1,...,T$}
\STATE Predict $\omega_t$. Consider the loss
function $\ell_t:\mathbb{R}^{K-1}\to\mathbb{R}$ given by measurement $\mathcal{M}_t:
\ell_t(\Tr(E_{t,1}\varphi),...,\Tr(E_{t,K-1}\varphi))$. It has the same value with the loss function defined in Eq.~\eqref{eq:DefLoss}. Let $\partial\ell_t/\partial x_j$ be a sub-derivative of $\ell_t$ with respect to $x_j$ for $j\in[K-1]$. Define
\begin{align}\label{eq:RFTLLoss}
\nabla_t:=\sum_{j=1}^{K-1}\frac{\partial\ell_t}{\partial (\Tr(E_{t,j})\omega_t)}E_{t,j}.
\end{align}

\STATE Update decision according to the RFTL rule with von Neumann entropy:
\begin{equation} \label{eq:UpdateRule}
\omega_{t+1}:=\arg\min_{\varphi\in\mathbb{C}^{d\times d}}\left\{\eta\sum_{s=1}^t\Tr(\nabla_s\varphi)+\sum_{i=1}^{d}\lambda_i(\varphi)\log\lambda_i(\varphi) \right\},
\end{equation}
where $\lambda_i(A)$ denotes the $i$-th eigenvalue of Hermitian matrix $A\in\mathbb{C}^{d\times d}$
\ENDFOR
\end{algorithmic}
\end{algorithm}

Algorithm~\ref{algo:RFTLKShad} employs von Neumann entropy, which relates to the Matrix Exponentiated Gradient algorithm \cite{Tsuda2005Matrix}. We remark that the loss function defined in Eq.~\eqref{eq:RFTLLoss} of the RFTL algorithm is slightly different from the definition in Eq.~\eqref{eq:DefLoss} in that it takes a vector of $K-1$ entries instead of $K$ entries. This is because the input vectors in Eq.~\eqref{eq:DefLoss} is supposed to be a probability distribution such that the summation of all entries is $1$. Therefore, there are $K-1$ free parameters. We rewrite the loss function with an input vector containing only free entries as Eq.~\eqref{eq:RFTLLoss}. According to the definition of regret in Eq.~\eqref{eq:DefRegret}, we now provide the following regret bound on this RFTL algorithm.
\begin{theorem}\label{thm:OnlineRegretUpper}
Setting $\eta=\sqrt{\log d/8T}$, the regret $R_T$ of Algorithm~\ref{algo:RFTLKShad} is bounded by $4\sqrt{(2\log 2)T\log d}$. 
\end{theorem}
\begin{proof}
We mainly follow the template of the proof for Theorem 3 in Ref. \cite{Aaronson2018Online}, but there are some differences since the loss function is different. We first observe that the loss function $\ell_t(\Tr(E_{t,1}\varphi),...,\\\Tr(E_{t,K-1}\varphi))$ is convex. There are at most two terms that contain each $\Tr(E_{t,j}\varphi)$ in the loss function when calculating the sub-derivative over each value $\Tr(E_{t,j}\varphi)$:
\begin{itemize}
    \item The variance in the $j$-th entry: $1/2\abs{\Tr(E_{t,j}\varphi)-b_{t,j}}$;
    \item The variance in the last entry: $1/2\abs{\Tr(E_{t,K}\varphi)-b_{t,K}}$ as $\Tr(E_{t,K}\varphi)=1-\sum_{j=1}^{K-1}\Tr(E_{t,j}\varphi)$.
\end{itemize}
Therefore, the value of sub-derivative $\partial\ell_t/\partial (\Tr(E_{t,j}))$ is either $\pm 1$ or $0$. We can divide all indexes $j$ of $E_{t,j}$ into three subsets $S_{t,1},S_{t,-1}$, and $S_{t,0}$ such that the value of $\partial\ell_t/\partial (\Tr(E_{t,j}))$ is $1,-1$, and $0$ for $j$ chosen from $S_{t,1},S_{t,-1}$, and $S_{t,0}$. We thus rewrite $\nabla_t$ as:
\begin{align*}
\nabla_t=\sum_{j\in S_{t,1}}E_{t,j}-\sum_{j\in S_{t,-1}}E_{t,j}.
\end{align*}
Notice that $E_{t,j}$ are projectors corresponding to different measurement outcomes and $\sum_{j=1}^KE_{t,j}=\mathbb{I}$, each $E_{t,j}$ are orthogonal and the spectral norm of any summation $\norm{\sum_{j\in[K]}E_{t,j}}\leq 1$. We can thus bound the spectral norm of $\nabla_t$ below by 
\begin{align*}
\norm{\nabla_t}\leq\norm{\sum_{j\in S_{t,1}}E_{t,j}}+\norm{\sum_{j\in S_{t,-1}}E_{t,j}}\leq 2.
\end{align*} 

In the following, we denote $\bm{\mu}_t=\Tr(E_{t,1}\omega_t),...,\Tr(E_{t,K-1}\omega_t)$ and $\bm{\tau}_t=\Tr(E_{t,1}\varphi),...,\Tr(E_{t,K-1}\varphi)$ for simplicity. Since $\ell_t$ is convex,
\begin{align*}
\ell_t(\bm{\mu}_t)-\ell_t(\bm{\tau}_t)\leq\nabla_t\cdot(\omega_t-\varphi)
\end{align*}
holds for all $\varphi\in\mathbb{C}^{d\times d}$, where $\cdot$ denotes the trace inner-product between complex matrices. Summing over $t$, we obtain
\begin{align*}
\sum_{t=1}^T[\ell_t(\bm{\mu}_t)-\ell_t(\bm{\tau}_t)]\leq\sum_{t=1}^T[\Tr(\nabla_t\omega_t)-\Tr(\nabla_t\varphi)].
\end{align*}
We define $g_t(X)=\nabla_t\cdot X$ for $X\in\mathbb{C}^{d\times d}$ and $H(X)$ to be the negative von Neumann Entropy of $X$. By Lemma 5.2 in Ref. \cite{Hazan2016Introduction}, we have
\begin{align}\label{eq:HazanLem5-2}
\sum_{t=1}^T[g_t(\omega_t)-g_t(\varphi)]\leq\sum_{t=1}^T\nabla_t\cdot(\omega_t-\omega_{t+1})+\frac{1}{\eta}D_R^2
\end{align}
for any $\varphi\in\mathbb{C}^{d\times d}$, where $D_R^2:=\max_{\varphi,\varphi'\in\mathbb{C}^{d\times d}}\{R(\varphi)-R(\varphi')\}$. We define $\Phi_t(X)=\eta\sum_{s=1}^t\nabla_s\cdot X+R(X)$, then line $5$ of Algorithm~\ref{algo:RFTLKShad} finds the minimal value of $\Phi_t(X)$ in $\mathbb{C}^{d\times d}$. To prove the theorem, we need the following two claims.
\begin{claim}
For all $t\in\{1.,,,.T\}$, we have $\omega_t\succeq 0$.
\end{claim}
\begin{proof}
Consider a Hermitian matrix $P\in\mathbb{C}^{d\times d}$ with zero minimal eigenvalue---i.e.,$\lambda_{\text{min}}=0$. Suppose $P=VQV^\dagger$, where $Q$ is a diagonal matrix with real entries as the eigenvalues of $P$. Assume $Q_{1,1}=\lambda_{\text{max}}(P)$ and $Q_{d,d}=\lambda_{\min}(P)=0$. We consider a different matrix $P'=VQ'V^\dagger$ such that $Q_{1,1}'=Q_{1,1}-\epsilon$, $Q_{i,i}'=Q_{i,i}$ for $i\in\{2,...,d-1\}$, and $Q_{d,d}'=\epsilon$ for $\epsilon<\lambda_{\max}(P)$. We then prove that there exists $\epsilon>0$ that satisfies $\Phi_t(P')\leq\Phi_t(P)$. By expanding both sides of the inequality, we need to prove an equivalent inequality
\begin{align*}
A\cdot(P'-P)\leq \alpha\log\alpha-(\alpha-\epsilon)\log(\alpha-\epsilon)-\epsilon\log\epsilon,
\end{align*}
where $A=\eta\sum_{s=1}^t\nabla_s$ and $\alpha=\lambda_{\max}(P)=Q_{1,1}$. Notice that $\norm{A}\leq\eta\sum_{s=1}^t\norm{\nabla_s}\leq2\eta t$. The left side of the inequality can be bounded using Generalized Cauchy-Schwartz inequality \cite{Bhatia2013Matrix} as
\begin{align*}
A\cdot(P-P')\leq2\eta t\norm{P-P'}_{\Tr}\leq4\epsilon\eta t.
\end{align*}
where $\norm{A}_{\Tr}$ is the trace norm for matrix $A$. As $\log\epsilon\to-\infty$ when $\epsilon\to0$, there exists a small enough $\epsilon$ such that $4\eta t\leq\log\alpha-\log\epsilon$. Therefore, we have
\begin{align*}
4\eta t\epsilon\leq\epsilon\log\alpha-\epsilon\log\epsilon\leq\alpha\log\alpha-(\alpha-\epsilon)\log(\alpha-\epsilon)-\epsilon\log\epsilon.
\end{align*}
This indicates that there exists $\epsilon$ that is small enough such that $\Phi_t(P')\leq\Phi_t(P)$. If $P$ has more than one zero eigenvalues, we can repeat the proof and construct the matrix $P'$. As $\omega_t$ is a minimal point of $\Phi_{t-1}$ and $\omega_1\succeq 0$, we have $\omega_t\succeq0$ for all $t$.
\end{proof}

Now, we can focus on $X\succeq 0$ and write $R(X)=\Tr(X\log X)$. We can further calculate the gradient of $\Phi_t(X)$ as
\begin{align*}
\nabla\Phi_t(X)=\eta\sum_{s=1}^t\nabla_s+\mathbb{I}+\log X.
\end{align*}
Here, we assume that the function $\Phi_t(X)$ is defined over real symmetric matrices. We can further prove the following claim.
\begin{claim}\label{clm:NonNegProd}
For all $t\in\{1,...,T-1\}$, $\nabla\Phi_t(\omega_{t+1})\cdot(\omega_t-\omega_{t+1})\geq0$.
\end{claim}
\begin{proof}
We inversely assume that $\nabla\Phi_t(\omega_{t+1})\cdot(\omega_t-\omega_{t+1})<0$. We choose a parameter $a\in(0,1)$ and construct $\overline{X}=(1-a)\omega_{t+1}+a\omega_t$. Then $\overline{X}\succeq0$ is also a density matrix. We denote $\Delta=\overline{X}-\omega_{t+1}=a(\omega_t-\omega_{t+1})$. According to Theorem 2 in Ref. \cite{Audenaert2005Continuity}, we have
\begin{align*}
\Phi_t(\overline{X})-\Phi_t(\omega_{t+1})&\leq a\nabla\Phi_t(\omega_{t+1})\cdot(\omega_t-\omega_{t+1})+\frac{\Tr(\Delta^2)}{\lambda_{\min}(\omega_{t+1})}\\
&=a\nabla\Phi_t(\omega_{t+1})\cdot(\omega_t-\omega_{t+1})+\frac{a^2\Tr((\omega_t-\omega_{t+1})^2)}{\lambda_{\min}(\omega_{t+1})}.
\end{align*}
Then we divide the above inequality by $a$ on both side and get
\begin{align*}
\frac{\Phi_t(\overline{X})-\Phi_t(\omega_{t+1})}{a}\leq\nabla\Phi_t(\omega_{t+1})\cdot(\omega_t-\omega_{t+1})+\frac{a\Tr((\omega_t-\omega_{t+1})^2)}{\lambda_{\min}(\omega_{t+1})}.
\end{align*}
Since we assume that $\nabla\Phi_t(\omega_{t+1})\cdot(\omega_t-\omega_{t+1})<0$, we can always choose some small enough $a$ such that the right hand side is negative while the left side is always positive since $\Phi_t(\overline{X})>\Phi_t(\omega_{t+1})$. This lead to an contradiction. Therefore, we have proved that $\nabla\Phi_t(\omega_{t+1})\cdot(\omega_t-\omega_{t+1})\geq0$.
\end{proof}

We define
\begin{align*}
B_{\Phi_t}(\omega_t||\omega_{t+1}):=\Phi_t(\omega_t)-\Phi_t(\omega_{t+1})-\nabla\Phi_t(\omega_{t+1})\cdot(\omega_t-\omega_{t+1}).
\end{align*}
By Pinsker inequality \cite{Carlen2014Remainder}, we have
\begin{align*}
\frac12\norm{\omega_t-\omega_{t+1}}_{\Tr}^2\leq\Tr(\omega_t\log\omega_t)-\Tr(\omega_t\log\omega_{t+1})=B_{\Phi_t}(\omega_t||\omega_{t+1}).
\end{align*}
Using Claim~\ref{clm:NonNegProd} and $\Phi_{t-1}(\omega_t)\leq\Phi_{t-1}(\omega_{t+1})$, we have
\begin{align*}
B_{\Phi_t}(\omega_t||\omega_{t+1})&=\Phi_t(\omega_t)-\Phi_t(\omega_{t+1})-\nabla\Phi_t(\omega_{t+1})\cdot(\omega_t-\omega_[t+1])\\
&\leq\Phi_t(\omega_t)-\Phi_t(\omega_{t+1})\\
&=\Phi_{t-1}(\omega_t)-\Phi_{t-1}(\omega_{t+1})+\eta\nabla_t\cdot(\omega_t-\omega_{t+1})\\
&\leq\eta\nabla_t\cdot(\omega_t-\omega_{t+1}).
\end{align*}
Therefore,
\begin{align*}
\frac12\norm{\omega_t-\omega_{t+1}}_{\Tr}^2\leq\eta\nabla_t(\omega_t-\omega_{t+1}).
\end{align*}
By Generalized Cauchy-Schwartz inequality, we have
\begin{align*}
\nabla_t\cdot(\omega_t-\omega_{t+1})&\leq\norm{\nabla_t}\norm{\omega_t-\omega_{t+1}}_{\Tr}\\
&\leq\norm{\nabla_t}\sqrt{2\eta\nabla\cdot(\omega_t-\omega_{t+1})}\\
&\leq2\eta\norm{\nabla_t}^2\\
&\leq8\eta.
\end{align*}
We combine this inequality with Eq.~\eqref{eq:HazanLem5-2} and reach the following bound
\begin{align*}
\sum_{t=1}^T\nabla_t\cdot(\omega_t-\varphi)\leq8\eta T+\frac{1}{\eta}D_R^2.
\end{align*}
We take $\eta=\frac{D_R}{2\sqrt{2T}}$. Observe that $D_R^2\leq \log d$ according to the definition of von Neumann entropy, the value for $\eta$ is
\begin{align*}
\eta=\sqrt{\frac{\log d}{8T}}.
\end{align*}
The corresponding regret bound is
\begin{align*}
\sum_{t=1}^T[\ell_t(\bm{\mu}_t)-\ell_t(\bm{\tau}_t)]\leq\sum_{t=1}^T\nabla_t\cdot(\omega_t-\varphi)\leq 4\sqrt{2T\log d}.
\end{align*}
\end{proof}

Now, we begin to prove \Cref{thm:OnlineUpper}. We consider the case that the RFTL is triggered when the prediction $\bm{\mu}_t=(\Tr(E_{t,1}\omega_t),...,\Tr(E_{t,K}\omega_t))$ deviates from the actual probability distribution $\bm{p}_t=(\Tr(E_{t,1}\rho),...,\Tr(E_{t,K}\rho))$ for more than $3\epsilon/4$---i.e.,$d_{TV}(\bm{\mu}_t,\bm{p}_t)>3\epsilon/4$. As the provided distribution $\bm{b}_t$ satisfies $d_{TV}(\bm{b}_t,\bm{p}_t)\leq\epsilon/4$, the loss function $\ell_t$ is at least $\epsilon/2$ by triangle inequality. 

We then consider using the real distribution in each iteration, the loss function is at most $\epsilon/4$ in each iteration. By the regret bound, we have
\begin{align*}
\frac{\epsilon}{2}T\leq\frac{\epsilon}{4}T+4\sqrt{2T\log d}.
\end{align*}
Therefore, we can obtain the upper bound on $T$ as $T\leq O(\log d/\epsilon^2)$. 

\subsection{Online Shadow Tomography of $K$-outcome POVMs}\label{sec:shadowComplexity}

We now prove \Cref{thm:MainUpper} using \Cref{thm:QThresKUpper} and \Cref{thm:OnlineUpper}. We describe our online shadow tomography procedure of $K$-outcome POVMs below.

Given the requirement parameters $\epsilon,\delta$ and the number of measurements $M$, we first define the following ancillary parameters
\begin{align*}
T_0=\left\lceil\frac{C_0\log d}{\epsilon^2}\right\rceil+1,\enspace\delta_0=\frac{\delta}{2T_0},\enspace N_0=\frac{C_1K\log(1/\delta_0)}{\epsilon^2}\log^2M,\enspace N_b=\frac{C_2K\log(1/\delta_0)}{\epsilon^2}\log^2M
\end{align*}
where $C_0,C_1$, and $C_2$ are three parameters that scale at most $\text{poly}(\log\log M,\log\log D,\log(1/\epsilon),\log K)$. The number of copies of $\rho$ will be $N=T_0(N_0+N_b)$, which is indeed
\begin{align*}
N=\tilde{O}\left(\frac{\log(1/\delta)}{\epsilon^4}\cdot K\cdot\log^2 M\log d\right),
\end{align*}
where $\tilde{O}$ hides a $\text{poly}(\log\log M,\log\log D,\log(1/\epsilon),\log K)$ factor.

After receiving $N$ copies of $\rho$, our algorithm first divide these states equally into $T_0$ batches, each consisting $N_0$ states. We prepare two joint states $\rho^{\otimes N_0}$ and $\rho^{\otimes N_b}$ using each batch. Each batch is used for the update procedure in a ``bad iteration" in our online learning procedure.

To begin with, the learner initializes the hypothesis state $\omega_0=\mathbb{I}/d$. In each iteration $t$, it chooses a fresh batch of states and runs the quantum distribution threshold search algorithm using joint state $\rho^{\otimes N_0}$. The threshold is chosen to be the probability distribution $\bm{\mu}_i$ after applying $\mathcal{M}_i$ for $i\in[M]$ on the hypothesis $\omega_t$. According to \Cref{thm:QThresKUpper}, we can always find such $C_1$ to solve this quantum distribution search problem with success probability at least $1-\delta_0$.

If the quantum distribution threshold search declares that for all $i\in[M]$, $d_{TV}(\bm{\mu}_i,\bm{p}_i)\leq\epsilon$. Then we have successfully found a hypothesis such that the probability distributions after applying all $K$-outcome POVMs on this hypothesis are at most $\epsilon$ from that of the unknown state $\rho$.

If the quantum distribution threshold search outputs $i^*$ where $d_{TV}(\bm{p}_{i^*},\bm{\mu}_{i^*})>3\epsilon/4$. We use $\rho^{\otimes N_b}$ for an estimation $\bm{b}_{i^*}$ for the probability distribution after applying $\mathcal{M}_{i^*}$ on $\rho$. According to Eq.~\eqref{eq:TotVarDisSampBound}, we can always find $C_2$ such that with probability at least $1-\delta_0$, one can bound the total variance distance $d_{TV}(\bm{p}_{i^*},\bm{b}_{i^*})\leq\epsilon/4$. We supply this $\bm{b}_{i^*}$ to the learner and the learner employs the Algorithm~\ref{algo:RFTLKShad} to update the hypothesis state into $\omega_{t+1}$. Furthermore, the remaining copies in the current batch will be abandoned. The learner will use a new batch and move into the next iteration.

According to \Cref{thm:OnlineUpper}, the number of "bad iterations" is bounded by $O(\log d/\epsilon^2)$. If there is no failure in all the rounds, we can always find $C_0$ such that we can guarantee that for all $i\in[M]$, $d_{TV}(\bm{\mu}_i,\bm{p}_i)\leq\epsilon$ after the online procedure, where $\bm{\mu}_i$ is obtained by applying $\mathcal{M}_i$ for $i\in[M]$ on the hypothesis $\omega_{T_0}$. Now we calculate the failure probability in this procedure. In each iteration, the success probability for the quantum distribution threshold search and the calculation of $\bm{b}_{i^*}$ are both at least $1-\delta_0$. By the union bound, the probability for failure after $T_0$ iterations is bounded by $2T_0\delta_0=\delta$. 

\subsection{An Exemplary Application}
Here, we provide some applications of our shadow tomography procedure of $K$-outcome POVMs. In quantum mechanics, we are sometimes interested in the expectation value of quantum operators $\{O_i\}_{i=1}^M$:
\begin{align*}
o_i=\expval{O_i}=\Tr(O_i\rho),
\end{align*}
given an unknown quantum state $\rho$. Suppose we perform a quantum measurement $\mathcal{M}_i$ that has $K$ outcomes to estimate the expectation value $o_i$. Then the following corollary holds by using our shadow tomography procedure
\begin{corollary}\label{coro:operator}
We consider an unknown $d$-dimensional quantum state, as well as $M$ quantum operators $O_1,...,O_M$. Assume we can measure each operator $O_i$ using a quantum measurement $\mathcal{M}$ of $K$ results. Then there exists a strategy that can approximate the expectation of each operator $\Tr(O_i\rho)$ within additive error $\epsilon$ using
\begin{align*}
N=\tilde{O}\left(\frac{\max_i\norm{O_i}^4}{\epsilon^4}\cdot K\cdot \log^2 M\log d\right)
\end{align*}
copies of $\rho$. Here, $\norm{\cdot}$ is the spectral norm. The success probability is at least $1-\delta$. 
\end{corollary}

To prove this corollary, we can divide the procedure into two steps. 

In the first step, we approximate the distribution after we apply each measurement $\mathcal{M}_i$ within total variance distance $\epsilon/\max_i\norm{O_i}$, which requires $N$ copies of $\rho$ according to \Cref{thm:MainUpper}.

Next, we calculate the expectation value using the distribution we obtained. The additive error for the expectation of $O_{i'}$ is bounded above by
\begin{align*}
\norm{O_{i'}}\cdot\frac{\epsilon}{\max_i\norm{O_i}}\leq\epsilon.
\end{align*}

As an example, we consider a $n$-qubit quantum states that is $d=2^n$-dimensional. We want to measure the expectation value for the operators $\{S_{\hat{n}_i}\}_{i=1}^M$ which measures the spin along $\hat{n}_i$ directions as
\begin{align*}
S_{\hat{n}_i}=\sum_{k=1}^n\sigma_{\hat{n}_i}^k\bigotimes_{k'\neq k}\mathbb{I}^k
\end{align*}
where $\sigma_{\hat{n}_i}^k$ denotes the spin operator along $\hat{n}_i$ on the $k$-th operator and $\mathbb{I}^k$ denotes the identity operator on the $k$-th qubit. Each measurement $\mathcal{M}_{i}$ has $K=n+1$ outcomes. The quantum event corresponding to each outcome $n-2k$ for $k=0,1,...,n$ can be written as a projector
\begin{align*}
A_{n-2k}=\sum_{x\in\{0,1\}^n\atop\abs{x}=x-2k}\ket{x}\bra{x},
\end{align*}
where $\abs{x}$ represents the Hamming weight for string $x$. We can calculate the spectral norm $\norm{\cdot}$ and the Hilbert-Schmidt norm $\norm{\cdot}_{\text{HS}}$ of $S_{\hat{n}_i}$ by
\begin{align*}
&\norm{S_{\hat{n}_i}}_{\text{HS}}=n2^n,\\
&\norm{S_{\hat{n}_i}}=n.
\end{align*}
Therefore, we can approximate the expectation value for $\{S_{\hat{n}_i}\}_{i=1}^M$ using 
\begin{align*}
N=\tilde{O}\left(\frac{\log^7 d}{\epsilon^4}\cdot \log^2 M\right)
\end{align*}
copies of $\rho$ according to Corollary~\ref{coro:operator}, which scales only poly-logarithmic on $d$. However, directly using classical shadow exponential number of samples.

\section{The Lower Bound}\label{sec:lower}
We now show that any shadow tomography procedure of $K$-outcome POVMs requires at least $\Omega(\min\{D^2,K+\log(M)\}/\epsilon^2)$ copies of $\rho$. It is worthwhile to mention that even in the classical special case when $M=1$, the following result (see, for example, \cite{Lee2021Lecture}) still shows that $\Omega(K)$ samples are required to estimate the probability distribution.
\begin{lemma}\label{lem:dislowTr}
Suppose we are given an unknown probability distribution $\mathcal{D}$ over set $\{1,...,K\}$. We use $N$ samples $x_1,...,x_N$ drawn from $\mathcal{D}$ to obtain an approximation $\mathcal{D}'$. Any approximation algorithm that gives an approximation $\mathcal{D}'$ such that $d_{TV}(\mathcal{D},\mathcal{D}')\leq\epsilon$ with probability at least $1-\delta$ requires sample complexity at least
\begin{align*}\label{eq:dislowTr}
\Omega\left(\frac{K+\log(\frac{1}{\delta})}{\epsilon^2}\right).
\end{align*}
\end{lemma}

Now, we prove the general ``quantum" case. We first need the following Lemma III.5 from Ref. \cite{Hayden2006Aspects}. 
\begin{lemma}\label{lem:Haydenspace}
(Hayden, Leung, and Winter \cite{Hayden2006Aspects}) Let $S$ and $T$ be subspaces of $\mathbb{C}^{d\times d}$ with dimension $d_1$ and $d_2$. We denote $\mathbb{P}_S$ and $\mathbb{P}_T$ to be projectors on subspaces $S$ and $T$. Consider $\rho_S=\frac{1}{d_1}\mathbb{P}_S$ to be the maximally mixed state projected onto $S$. If we fix $T$ and randomly choose $S$, then
\begin{align*}
\Pr\left[\abs{\Tr(\mathbb{P}_T\rho_S)-\frac{d_2}{d}}\geq\frac{c_0d_2}{d}\right]\leq\exp(-\cdot\frac{c_0^2d_1d_2}{6\ln2}).
\end{align*}
\end{lemma}

Now, we begin to prove \Cref{thm:MainLower}. We set $D:=\lfloor\min\{d,\sqrt{\log_2 M+K}\}\rfloor$ and suppose the unknown state is $D$-dimensional mixed state. We choose some constant $c\in(0,1)$ and set $L=\lfloor c^{D^2-K}\rfloor$. We will have $L$ quantum measurements of $K$ outcomes for $L\leq M$. Notice that the probability for each outcome of the quantum measurement can be regarded as the expectation for a quantum event. There are in total $L\cdot K$ quantum events.

We choose $L\cdot K$ subspaces $\{S_{1,1},...,S_{1,K}\},...,\{S_{L,1},...,S_{L,K}\}$ from $\mathbb{C}^{D\times D}$ for $L$ POVMs independently and Haar-randomly such that $\dim(S_{i,j})=D/K$ for $i\in[L]$ and $j\in[K]$. The $K$ subspaces in each set $\{S_{i,1},...,S_{i,K}\}$ are orthogonal. We denote $\mathbb{P}_{i,j}$ to be the projection to $S_{i,j}$ and $\rho_{i,j}=K\mathbb{P}_{i,j}/D$ to be the maximally mixed state projected onto $S_{i,j}$. As long as we choose a $c$ that is close enough to $1$, we can always find a choice over $S_{i,j}$'s with success probability $1-o(1)$ such that
\begin{equation}\label{eq:Projectassumption}
\abs{\Tr(P_{i,j}\rho_{i',j'})-\frac{1}{K}}\leq\frac{1}{2K}
\end{equation}
for $i\neq i'$ according to Lemma~\ref{lem:Haydenspace}. We fix such a choice over $S_{i,j}$.

Without loss of generality, we assume that $K$ is even. Now, we consider constructing the following states using a classical bit string $\bm{z}=(z_1,...,z_{K/2})$ of $K/2$ bits
\begin{align*}
\rho_i(\bm{z}):=\sum_{j=1}^{K/2}\left[\frac{1-50\epsilon z_j}{K}\rho_{i,2j-1}+\frac{1+50\epsilon z_j}{K}\rho_{i,2j}\right].
\end{align*}
We consider applying measurement $\mathcal{M}_i$ on state $\rho_i(\bm{z})$. The $(2j-1)$-th and the $2j$-th entry for the probability distribution is
\begin{align*}
&\Tr(E_{i,2j-1}\rho_i(\bm{z}))=\frac{1-50\epsilon z_i}{K},\\
&\Tr(E_{i,2j}\rho_i(\bm{z}))=\frac{1+50\epsilon z_i}{K}.\\
\end{align*}
Therefore, the probability distribution can be written as
\begin{align*}
\left(\frac{1-50\epsilon z_1}{K},\frac{1+50\epsilon z_1}{K},...,\frac{1-50\epsilon z_{K/2}}{K},\frac{1+50\epsilon z_{K/2}}{K}\right).
\end{align*}
We consider applying $\mathcal{M}_{i'}$ on state $\rho_i(\bm{z})$ for $i\neq i'$. According to Eq.~\eqref{eq:Projectassumption}, the $j$-th entry of the probability distribution is
\begin{align*}
\Tr(E_{i',j}\rho_i(\bm{z}))&\leq\frac34\cdot\frac{1+50\epsilon}{K}+\frac{1}{4}\cdot\frac{1-50\epsilon}{K}=\frac{1+25\epsilon}{K},\\
\Tr(E_{i',j}\rho_i(\bm{z}))&\geq\frac34\cdot\frac{1-50\epsilon}{K}+\frac{1}{4}\cdot\frac{1+50\epsilon}{K}=\frac{1-25\epsilon}{K}.
\end{align*}

Now, we fix a measurement $\mathcal{M}_i$. If we apply this measurement to two quantum states $\rho_i({\bm{z}_1})$ and $\rho_{i'}(\bm{z}_2)$ for $i\neq i'$. The total variance distance for the two probability distribution is at least $25\epsilon/2$. It follows that, if we can estimate the probability distribution after a $\mathcal{M}_i$ to within total variance distance $\epsilon$, we can immediately estimate $i\in[L]$ for the unknown state $\rho_i(\bm{z})$.

We then consider applying this measurement to two quantum states $\rho_i({\bm{z}_1})$ and $\rho_i(\bm{z}_2)$ for $\bm{z}_1\neq\bm{z}_2$. Since each single difference on one entry in $\bm{z}_1$ and $\bm{z}_2$ will contribute $\frac{100\epsilon}{n}$ to the total variance distance between the two probability distribution, the distance is at least $\epsilon$ if more than $1\%$ of the entries are different. Therefore, we can distinguish between such $\bm{z}_1$ and $\bm{z}_2$ if we can estimate the probability distribution after a $\mathcal{M}_i$ to within total variance distance $\epsilon$. 

Suppose we choose $i$ and $\bm{z}$ uniformly at random, then it contains $\log_2(2^{K/2}M)=\Omega(K+\log_2(M))$ bits of classical information. Suppose we require $N$ copies of $\rho$ to perform a shadow tomography procedure of $K$-outcome POVMs. Let
\begin{align*}
\zeta:=\mathbb{E}_{i\in[L],\bm{z}\in\{0,1\}^{K/2}}\left[\rho_i(\bm{z})^{\otimes N}\right].
\end{align*}
In order to make learning $i$ and $99\%$ of the entries for $\bm{z}$ from $\zeta$ information-theoretically possible, the mutual information $I(\zeta:i,\bm{z})$ must be at least $\Omega(K+\log_2(M))$. As both $i$ and $\bm{z}$ are classical, we have
\begin{align*}
I(\zeta:i,\bm{z})&=S(\zeta)-S(\zeta|i)\\
&=S(\zeta)-S(\rho_i(\bm{z})^{\otimes N})\\
&\leq N(\log_2 D-S(\rho_i(\bm{z})),
\end{align*}
where $S(\cdot)$ is the von Neumann entropy. Now, we calculate the term $S(\rho_i(\bm{z}))$. Let $\lambda_{i,\bm{z},1},...,\lambda_{i,\bm{z},D}$ be the eigenvalues for $\rho_i(\bm{z})$. By applying a unitary transformation that diagonalizes $\rho_i(\bm{z})$ rotating to a basis that contains half of the projectors, we can observe that half the $\lambda_{i,\bm{z},j}$'s are $(1+50\epsilon)/D$ and the other half of the $\lambda_{i,\bm{z},j}$'s are $(1-50\epsilon)/D$. Hence,
\begin{align*}
S(\rho_i(\bm{z}))&=\sum_{j=1}^D\lambda_{i,\bm{z},j}\log_2(\frac{1}{\lambda_{i,\bm{z},j}})\\
&=\frac D2\cdot\frac{1-50\epsilon}{D}\log(\frac{D}{1-50\epsilon})+\frac D2\cdot\frac{1+50\epsilon}{D}\log(\frac{D}{1+50\epsilon})\\
&=\log_2D-\left(1-\frac{1-50\epsilon}{2}\log(\frac{2}{1-50\epsilon})-(1+\frac{1+50\epsilon}{2}\log(\frac{2}{1+50\epsilon})\right)\\
&\geq\log_2 D-O(\epsilon^2).\\
\end{align*}
Therefore, the mutual information can be bounded by
\begin{align*}
I(\zeta:i,\bm{z})=O(N\epsilon^2).
\end{align*}
As learning $i$ and $99\%$ of the entries for $\bm{z}$ requires $\Omega(K+\log_2(M))$ bits of classical information, we conclude that
\begin{align*}
N=\Omega\left(\frac{D^2}{\epsilon^2}\right)=\Omega\left(\frac{\min\{d^2,K+\log(M)\}}{\epsilon^2}\right).
\end{align*}

\section{Open Problems}\label{sec:openq}
This work established the exact dependence on $K$ for shadow tomography of $K$-outcome quantum measurements and proposed the explicit algorithm that learns these distributions with sample complexity optimal in $K$. But there are further problems related to learning distributions over quantum measurement outcomes. Here are a few possible future directions.

\textbf{Tight Bounds on $M$ and $d$. }What is the true sample complexity dependence on $M$ and $d$? For a shadow tomography procedure of $K$-outcome, the best known lower bound is $\Omega(\min\{d^2,K+\log M\}/\epsilon^2)$. In particular, in the case of $K=2$, lower bounds with different constraints have been developed. Ref. \cite{Aaronson2019Gentle} connects differential privacy, gentleness and shadow tomography. It was shown that if a shadow tomography procedure is gentle on product states, then at least $\Omega(\log M\sqrt{\log d}/\epsilon^2)$ copies of unknown quantum data is needed. On the other hand, if a shadow tomography procedure is online, it needs $\Omega(\sqrt{\min\{M,\log d\}})$ copies of unknown quantum data \cite{Ullman2018Limits}. These results, however, are both classical results \cite{Bun2018fIngerprinting,Vadhan2017Complexity}. If we drop these assumptions on the shadow tomography procedure, the lower bound will be independent of $d$ \cite{Aaronson2019Shadow,Huang2020Predicting}. 

(1) A first question is whether we can develop a better upper bound---i.e., find an algorithm with sample complexity that has smaller dependence on $\log M$, $\log d$, and $1/\epsilon$. While the current best bound for shadow tomography is $O(\log^2M\log d/\epsilon^4)$, it has been argued that a classical match for this problem is the task known as \textit{Adaptive Data Analysis} \cite{Dwork2015Preserving}. The task is known to be solved using $\tilde{O}(\log M\sqrt{\log d}/\epsilon^3)$, which is optimal in the dependence on $M$ and $d$.

(2) An alternative question is whether we can find a better lower bound on the number of unknown quantum data. Specifically, when we assume the shadow tomography procedure is gentle or online, the current lower bounds are obtained in a classical context. Is it possible to find a better ``quantum" bound for gentle or online shadow tomography?

\textbf{Shadow Tomography with State-preparation Unitaries. }

(3) In the setting of shadow tomography, we assume that we are given $N$ copies of unknown states $\rho$. We then perform measurements on $\rho^{\otimes N}$ to evaluate the expectation value of the given observables. In some cases, we can obtain the state-preparation unitaries and perform unitary oracles on them. Following this vein, Huggins et al.\ \cite{Huggins2021Nearly} proposed an optimal algorithm of query complexity $O(\sqrt{M}/\epsilon)$ for the state-preparation oracle to estimate $M$ observables \textit{independently} within additive error $\epsilon$. van Apeldoorn et al.\ \cite{Van2022Quantum} proved that state tomography problem can be solved using $O(dr/\epsilon)$ queries. It would be interesting to know the upper and lower bound for quantum shadow tomography under this stronger input assumption \footnote{At least we can use the technique of quantum mean estimation on multivariate random variables \cite{Hamoudi2021Quantum,Cornelissen2022Near} to reduce the dependence on $\epsilon$ from $\epsilon^{-4}$ to $\epsilon^{-3}$.}.

\textbf{Shadow Tomography with Restricted Kinds of Measurements. }

(4) From the perspective of experimental feasibility, our algorithm requires joint measurements, which measure the state $\rho^{\otimes N}$ simultaneously, and weak measurements (or the so-called \textit{non-demolition measurements}), which carefully maintain the state through a long sequence of quantum measurements. Therefore, it would be interesting if we can find a shadow tomography of $K$-outcome POVMs using incoherent measurements---i.e., measuring each quantum state separately. It would also be meaningful if we can find a shadow tomography procedure that measures each quantum state at most once.

(5) In this work, we prove the tight bound that $\Theta(K)$ copies of $\rho$ are required to learn the distribution after quantum measurements. However, $K$ can be exponentially large in some cases. For example, if we perform a measurement on the computational basis, $K=d=2^n$. For a fixed dimension $d$, can we find a class of quantum measurements that have a (sub)polynomial number of outcomes while only requiring logarithmic sample complexity on $d$?

\textbf{Computational Complexity. }

(6) Our shadow tomography procedure takes poly$(M,d,1/\epsilon,K)$ time complexity to compute the estimation for the distributions. Although it was proved that if there exists a ``hyperefficient" shadow tomography procedure---i.e., the sample complexity and sample complexity are both poly $(\log M,\log d,1/\epsilon,K)$, then any quantum advice can be simulated by classical advice and $\mathsf{BQP/qpoly}=\mathsf{PostBQP/poly}$ \cite{Aaronson2019Shadow}. It would be interesting if we can find a shadow tomography procedure that is hyperefficient when we add some constraints on $\mathcal{M}_i$ and $\rho$. Also, can we find a trade-off relation between the sample and the time complexity?

\section{Acknowledgments}
We thank William Kretschmer, Dong-Ling Deng, and Weikang Li for helpful discussions and comments.

\bibliographystyle{plain}
\bibliography{KPOVMShadow}
\newpage
\appendix  

\end{document}